%% file: cccg20.tex
\documentclass{cccg20}
\usepackage{graphicx,amssymb,amsmath}

\usepackage{tda}

\title{If You Must Choose Among Your Children, Pick the Right One}

\author{
    Brittany Terese Fasy\footnotemark[2]~\thanks{Department of Mathematical
    Sciences, Montana State U.}
        \and
        Benjamin Holmgren\thanks{School of Computing,
            Montana State
            U. \newline
   {\scriptsize {\tt \{brittany.fasy, bradleymccoy, david.millman\}@montana.edu}
    {\tt benjamin.holmgren@student.montana.edu}}}
        \and
        Bradley McCoy\footnotemark[2]
        \and
        David L. Millman\footnotemark[2]
    }

\index{Fasy, Brittany Terese}
\index{Holmgren, Benjamin}
\index{McCoy, Bradley}
\index{Millman, David L.}

\begin{document}
\thispagestyle{empty}
\maketitle

\begin{abstract}
    \input{body/abstract}

\end{abstract}

\input{body/introduction}

\input{body/background}

\input{body/king}
\input{body/algorithm}

\input{body/discussion}

\paragraph{Acknowledgements}
\input{body/acknowledgements}




{
\small
\bibliographystyle{abbrv}
\bibliography{references}
}

\appendix
\section{Additional Details for \extract}
\input{body/append-extract}

\camera{\input{body/append-defs}}
\camera{\input{body/append-link}}

\end{document}

%% file: body/abstract.tex
Given a simplicial complex $K$ and an injective function~$f$ from the vertices
of~$K$ to $\mathbb{R}$,
 we consider algorithms that extend $f$ to a discrete Morse function on~$K$. 
 We show that an algorithm of King, Knudson and Mramor can be described 
 on the directed Hasse diagram of~$K$.
 Our description has a faster runtime for high dimensional data with no increase in space.

%% file: body/introduction.tex
\section{Introduction}
\label{sec:intro}

Milnor's classical Morse theory provides tools for investigating the
topology of smooth manifolds \cite{milnor63}.  In \cite{forman98}, Forman showed
that many of the tools for continuous functions can be applied in the discrete
setting.  Inferences about the topology of a CW complex can be made from the number
of critical cells in a Morse function on the complex.

Given a Morse function one can interpret the function in many ways. Switching interpretations is often revealing.
In this paper, we think of a discrete Morse function in three different ways.
Algebraically, a Morse function is a function from the faces of a complex to the real numbers,
subject to certain inequalities. Topologically,
a Morse function is a pairing of the faces such that the removal of any pair
does not change the topology of the complex. Combinatorially,
a Morse function is an acyclic matching in the Hasse diagram of the complex, where unmatched
faces correspond to critical cells.

Discrete Morse theory can be combined with persistent homology to
analyze data, see
\cite{king,uli-phd,bauer2012optimal,edelsbrunner02,wang18,vcomic2011dimension,edelsbrunner2003hierarchical}.
When dealing with data, we have the additional constraint
that vertices have function values assigned.
For complexes without any preassigned function values,
Joswig and Pfetsch showed
that finding a Morse function with a minimum number of critical cells is
NP-Hard \cite{joswig04}.  Algorithms that find Morse functions with relatively
few critical cells have been explored in
\cite{lewiner03,nanda2013morse,hersh05}.

In this work, we consider the algorithm \extract, \algref{extract},
given in \cite{king}.  \extract \ takes as input a simplicial complex and an
injective function from the vertices to the reals, and returns a discrete Morse
function, giving topological information about the complex. We show that a
subalgorithm of \extract, \eraw\ can be simplified by considering the directed Hasse
diagram. This simplification leads to an improved runtime and no change in
space.  The paper is organized as follows, in \secref{background}, we provide
the definitions that will be used in the paper. In
\secref{king}, we describe \extract\ and analyze the runtime, then, in
\secref{ours}, we give our reformulation and show that the runtime is improved
from \erawrun{} to \ercrun{} where $n$ is the number of
cells and $d$ is the dimension of $K$.

%% file: body/background.tex
\section{Background}
\label{sec:background}

In this section, we provide definitions, notation, and primitive operations
used throughout the paper.
For a general overview of discrete Morse theory see~\cite{nic,knudson2015morse},
note that both texts provide a description of
\extract\ originally given in~\cite{king}. \extract\ is the starting point for this work.

In what follows, we adapt the notation of Edelsbrunner and
Harer~\cite{edelsbrunner2010computational} to the
definitions of Forman~\cite{forman02}.
Here, we work with simplicial complexes, but the results hold for CW complexes.
Let $K$ be a simplicial complex with $n$ simplices.
For $i \in \N$, denote the $i$-simplices of $K$ as $K_i$,
the number of simplices in $K_i$ as $n_i$,
and  the dimension of the highest dimensional simplex of $K$ as $\dim{K}$.

Let $\sigma \in K,$ denote the dimension of $\sigma$ as $\dim{\sigma}$.
Let $p = \dim{\sigma}$
and $\{v_0, v_1,\cdots,v_p \} \subseteq K_0$ be the zero-simplices of $\sigma$,
then we say $\sigma=[v_0,v_1,\cdots,v_p]$.
If $\tau \in K$ is disjoint from $\sigma$,
then we can define the \emph{join} of $\sigma$ and $\tau$ to be the
$(\dim{\sigma}+\dim{\tau}+1)$-simplex that
consists of the union of the vertices in $\sigma$ and $\tau$,
denoted $\sigma * \tau$.
We write $\sigma \prec \tau$ if $\sigma$ is a proper face of $\tau$.

Let $p \in \N$ and consider simplices $\sigma_u, \sigma_v \in K_p$, with
$\sigma_u = [u_0, u_1, \ldots, u_p]$ and
$\sigma_v = [v_0, v_1, \ldots, v_p]$.
Let $f_0:K_0 \rightarrow \R$ be an injective function.
Without loss of generality, assume that the zero-simplices of $\sigma_u$ and
$\sigma_v$ are sorted by function value, that is,
we have $f_0(v_{i})<f_0(v_{j})$ when $0 \leq i < j \leq p$,
similarly for $\sigma_u$.
We say that
$\sigma_u$ is lexicographically smaller than $\sigma_v$,
denoted $\sigma_u \lexlt \sigma_v$,
if the vector
$\langle f_0(u_p), f_0(u_{p-1}), \ldots, f_0(u_0) \rangle$
is lexicographically smaller than
$\langle f_0(v_p), f_0(v_{p-1}), \ldots, f_0(v_0) \rangle$.

The \emph{star} of $v$ in $K$, denoted $\starK{v}$,
is the set of all simplices of $K$ containing $v$.
The \emph{closed star} of $v$ in~$K$, denoted $\closedStarK{v}$,
is the closure of $\starK{v}$.
The \emph{link} of $v$ in $K$, is denoted as
$\linkK{v} := \closedStarK{v} \setminus \starK{v}$.
We define the \emph{lower link} of $v$, denoted $\llinkK{v}$, to be the maximal subcomplex
of $\linkK{v}$  whose zero-simplices have function value less than $f_0(v)$; the
lower link can be computed in~$O(n)$~time.
\camera{See \algref{link} for details on computing the lower link.}

We provide the definition of a Morse function, modified from Forman~\cite{forman02}.
\camera{See
\appendref{defs} for the equivalence of the definitions.}

\begin{restatable}[Morse Function]{definition}{morsedefn}\label{def:morsedefn}
A function $f : K\rightarrow \R$ is a \emph{discrete Morse function},
    if for every $\sigma \in K$,
    the following two conditions hold:
\begin{enumerate}
    \item $|\{ \beta \succ \sigma | f(\beta)\leq f(\sigma)\}|\leq 1,$
    \item $|\{ \gamma \prec \sigma | f(\gamma)\geq f(\sigma)\}|\leq 1.$
\end{enumerate}    \end{restatable}

An intuitive definition is given in \cite{nic},
``the function generally increases as you increase the dimension
of the simplices. But we allow at most one exception per simplex."
Let $f : K\rightarrow \R$ be a discrete Morse function.
A simplex $\sigma \in K$ is \emph{critical}
if the following two conditions hold:
\begin{enumerate}
    \item $|\{ \beta \succ \sigma | f(\beta)\leq f(\sigma)\}|=0,$
    \item $|\{\gamma \prec \sigma | f(\gamma)\geq f(\sigma)\}|=0.$
\end{enumerate}
Simplices that are not critical are called \emph{regular}.

Given a discrete Morse function $f$ on a simplicial complex $K$, we define
the induced \emph{gradient vector field}, or GVF, for short, as
$\{(\sigma,\tau) : \sigma \prec \tau, f(\sigma)\geq f(\tau)\}$.
Note that
$\sigma$ is a codimension one face of $\tau$.
\camera{See \lemref{impossible}.}
We can gain some intuition for this definition by drawing
arrows on the simplicial complex as follows.
If $\sigma$ is regular, a codimension one face of $\tau$, and
$f(\tau) \leq f(\sigma)$,
then we draw
an arrow from $\sigma$ to $\tau$.
Constructing a GVF for a \simp\
is as powerful as having a discrete Morse function, and is the goal of both
\extract\ and our proposed \algref{erchild}.

Next, we define two functions that are helpful when constructing a GVF.
The \emph{rightmost face} of $\sigma$, denoted~$\rchild{\sigma}$,
is the face of $\sigma$ with maximum lexicographic value.
The \emph{leftmost coface} of $\sigma$, denoted $\lparent{\sigma}$,
is the dimension one coface of $\sigma$ with minimum lexicographic value.
We say $\sigma$ is a \emph{left-right parent} and we call $\rchild{\sigma}$ a
\emph{left-right child} if
$\lparentName \circ \rchild{\sigma} = \sigma$.

In \cite{forman02}, Forman showed that each simplex in $K$ is exclusively a tail, head,
or unmatched.  Moreover, the unmatched simplices are critical.
Thus, we can partition the simplices of $K$ into
\emph{heads} $\heads$,
\emph{tails}~$\tails$, and
\emph{critical simplices} $\criticals$,
and encode the GVF as a bijection
$\matches: \tails \rightarrow \heads$.
That is, we can represent the GVF for $f$ as the unique tuple $\gvf$.
\camera{the uniquenes of the GVF doesn't seem to have a lemma that we can point
    to. Add something to appendix to show GVF is well-defined.}
We will use this representation throughout our algorithms.

Note that a GVF is a particularly useful construction.
It provides a way to reduce the size of a simplicial complex without changing
the topology (by cancelling matched pairs), which is constructive for preprocessing large 
simplicial complexes.
See \cite{nanda2013morse,vcomic2011dimension} for examples.

We define a consistent GVF as follows:
\begin{definition}[Consistent GVF]\label{def:consistent}
    Let $K$ be a simplicial complex, and let~$f_0 \colon K_0 \to \R$ be injective.
    Then, we say that a gradient vector field $\gvf$ is \emph{consistent} with $f_0$ if, for
    all~$\eps >0$, there exists a discrete
    Morse function~$f \colon K \to \R$ such~that
    \begin{enumerate}[(a)]
        \item $\gvf$ is the GVF corresponding to $f$.\label{def:consistent-gvf}
        \item $f|_{K_0}=f_0$.\label{def:consistent-restr}
        \item $|f(\sigma) - \max_{v \in \sigma}f_0(v)| \leq
            \eps$.\label{def:consistent-eps}
    \end{enumerate}
\end{definition}

Let $\gvf$ be a GVF.
Then, for $r,p \in \N$,
a \emph{gradient path}\footnote{
    There is a slight discrepancy between the definition of
    Forman~\cite{forman02} and KKM~\cite{king}.  In particular,
    Forman's definition states the head and the tail of the path are
    simplices of the same dimension.
    On the other hand, KKM's usage in the $\ecancel$ algorithm
    expects that the head and tail are different dimensions.
    Here, we state the definition implied by the usage in KKM.
}
is a sequence of simplices in $K$:
$$
    \Gamma = \{\sigma_{-1},
        \tau_0, \sigma_0,
        \tau_1, \sigma_1,
        \ldots,
    \tau_{r},\sigma_{r},
    \tau_{r+1} \}
$$
beginning and ending with
critical simplices $\sigma_{-1} \in K_{p+1}$ and $\tau_{r+1} \in K_{p}$
such that for $0 \leq i < r$,
$\tau_i \in K_p$,
$\sigma_i\in K_{p+1}$,
$\matches(\tau_i) = \sigma_i$,
and $\tau_i \succ \sigma_{i+1} \neq \sigma_i$.
We call a path \emph{nontrivial} if $r > 0$.

%% file: body/king.tex
\section{A Discrete Morse Extension of $f_0 \colon K_0 \to \R$}
\label{sec:king}

In this section, we give a description of
\algref{extract}~(\extract), originally from \cite{king}.
This algorithm takes a simplicial complex $K$,
an injective function $f_0 \colon K_0 \to \R$,
and a threshold that ignores pairings with small persistence $p \geq 0$;
and returns a GVF on $K$ that is consistent with $f_0$.
\begin{algorithm}
    \caption{\cite{king} \extract}\label{alg:extract}
    \begin{algorithmic}[1]
        \Require A \simp\ $K$, injective function
        $f_0\colon K_0 \to \R$, and $p\geq 0$
        \Ensure a GVF consistent with $f_0$

        \State $\gamma \gets \eraw(K,f_0)$
        \algorithmiccomment{\algref{eraw}}\label{line:extract-calleraw}

        \For{$j=1,2, \ldots,\dim{K}$}
        \State $\gamma \gets \ecancel(K,h,p,j,\gamma)$
        \algorithmiccomment{\shortalgref{ecancel}}

        \EndFor\\
        \Return $\gamma$ 
    \end{algorithmic}
\end{algorithm}

\extract\ uses two subroutines: First, in \lineref{calleraw}{extract} \eraw\ (given in
\algref{eraw}) is used to generate an initial GVF on $K$ consistent with
$f_0$.  Let $(H_0,T_0,C_0,r_0)$ be this initial GVF.
Then, for each dimension ($j = 1$ through $\dim{K}$), the algorithm makes a call to \ecancel\
(given in \algref{ecancel}) that augments an existing gradient path
 to remove simplices from $C_0$ in pairs. For more details, see
 \appendref{kingsubroutines}.

In the next section, we provide a simpler and faster algorithm to replace \eraw,
which dominates the runtime of $\extract$ when $p=0$
(and in practice, when $p$ is very small).  We
conclude this section with properties of the output from \eraw:

\begin{restatable}[Properies of \eraw]{theorem}{kingthm}\label{thm:eraw}\label{thm:king}
    Let $K$ be a simplicial complex, let $f_0 \colon K_0 \to \R$ be an injective
    function, and suppose~$\gvf$ is the output of $\eraw(K,f_0)$.
    Let~$\eps > 0$.
    Then, there exists a discrete Morse function~$f \colon K \to \R$ such that
    the following~hold:
    \begin{enumerate}[(i)]
        \item $\gvf$ is a GVF consistent with~$f_0$.\label{part:eraw-gvf}
        \item Let $\sigma \in K$.  Then, $\sigma \in \heads$ if and only if
            $\sigma$ is a left-right parent.\label{part:eraw-ifftails}
        \item For all $\sigma \in \heads$,
            $\matches(\rchild{\sigma})=\sigma$.\label{part:eraw-composition}
        \item The runtime of \eraw\ is \erawrun.\label{part:eraw-timebound}
    \end{enumerate}
\end{restatable}

%% file: body/algorithm.tex
\section{A Faster Algorithm for \eraw}
\label{sec:ours}

The main contribution of this paper is
\erchild, which we show is a simplified version of
\eraw{} that has the same output with an improved runtime.
This section provides a description of the algorithm, and a proof of the
equivalence with \eraw.

\subsection{Hasse Diagram Data Structure}
\input{body/hasse-dgm}

\subsection{Algorithm Description}
\input{body/algo-right-child}

\subsection{Analysis of \erchild}

For the remainder of this section, we prove that \algref{erchild} (\erchild) is
equivalent to and faster than \algref{eraw} (\eraw).  For the following lemmas,
let $K$ be a simplicial complex, let $f_0
\colon K_0 \to \R$ be an injective function, and
let~$\gvf$ be the output of \callerchild.

First, we show that $\gvf$ is a partition of $K$.
\begin{lemma}[Partition]\label{lem:partition}
    The sets $\heads$, $\tails$, and $\criticals$ partition $K$.
\end{lemma}

\begin{proof}
By \shortlineref{downd}{erchild} and \lineref{ltor}{erchild}, \erchild\ iterates over all $\sigma^d \in K$ with $d > 0$ once.
Each $\sigma$ is either assigned or unassigned.
If $\sigma$ is unassigned, there are two options; $\sigma$ may be a left-right parent, or it may not be.
If $\sigma$ is a left-right parent, \shortlineref{addh}{erchild} ensures that $\sigma$ is put into $\heads$.
Otherwise, \shortlineref{addc}{erchild} ensures that $\sigma$ is put into $\criticals$.
If $\sigma$ is assigned, then $\sigma$ was assigned to $\tails$ in  \shortlineref{addh}{erchild}.
Thus, every  $\sigma^d \in K$ with $d > 0$ must be assigned to exactly one of $\heads, \tails,$ or $\criticals$.
Then, every $\sigma^0$ is again either assigned or unassigned. If assigned, $\sigma^0\in \tails$.
 If unassigned, $\sigma^0$ is added to $\criticals$ in \shortlineref{leftovers}{erchild}.
Thus, every $\sigma \in K$ is assigned one of $\heads, \tails,$ or $\criticals$, making $\heads, \tails,$
  and $\criticals$ partition $K$.
\end{proof}

 We will show that $\gvf$ satisfies \partinref{eraw}{gvf}, \partinref{eraw}{ifftails}, and
   \partinref{eraw}{composition} of \thmref{eraw}. Later in this section, we show that any GVF with these properties is unique.
   
   First, we show \partinref{eraw}{composition} and one direction of \partinref{eraw}{ifftails}.

\begin{lemma}[Child Heads are Parents]\label{lem:wefindmatch}
    Let $\sigma \in \heads$. Then,
    $\sigma$ is a left-right parent and $\matches(\rchild{\sigma})=\sigma$.
\end{lemma}

\begin{proof}
    Recall that $\heads$ is the second output of
    \callerchild, given in \algref{erchild}.
    As \shortlineref{addh}{erchild} is the only step in which simplices are
    added to~$\heads$ and is within an $if$ statement
    that checks if $\sigma$ is a left-right parent, $\sigma$ must
    be a left-right parent.
    Also within the $if$ statement,
    \shortlineref{defro}{erchild} adds
    $(\rchild{\sigma},\sigma)$ to $\matches$,
    which means that $\matches(\rchild{\sigma}) = \sigma$.
\end{proof}

Now we show the reverse direction of \partinref{eraw}{ifftails}. 

\begin{lemma}[Child Parents are Heads]\label{lem:wefindtails}
    Let $\sigma \in K$.  If $\sigma$ is a left-right parent, then $\sigma \in
    \heads$.
\end{lemma}

\begin{proof}
    Recall that in order for $\sigma$ to be a left-right parent, we must have $\lparent{\rchild{\sigma}} = \sigma$.
    Now, we consider two cases. For the first case,
    suppose $\sigma \in C$. Then $\sigma$ is added to $C$
    in \lineref{addc}{erchild} when $\rchild{\sigma}$ must already be assigned to
    $h \lexlt \sigma$. So, $\lparent {\rchild{\sigma}}=h\neq \sigma$ and $\sigma$
    is not a left-right parent.

    For the second case, suppose $\sigma=[v_0,v_1,\ldots,v_d] \in \tails$.
    Then $\sigma$ is added
    to $\tails$ in \lineref{addh}{erchild} where $\sigma=\rchild{h}$
    for some $h=[v_{-1},v_0,\ldots, v_d]\in \heads$ with $f_0(v_{-1})<f_0(v_0)$.
    Notice that $\xi=[v_{-1},v_1,v_2,\ldots, v_d]$ is a face of $h$ and 
    $\xi \lexlt \sigma$.
     Then,
    $\lparent {\rchild{\sigma}}=\lparent{[v_1,\ldots,v_d]}\leq_{lex}[v_{-1},v_1,v_2,\ldots, v_d]
    \lexlt [v_0,v_1,\ldots,v_d]=\sigma$
    and $\sigma$ is not a left-right parent.

    Thus, if $\sigma$ is a left-right parent, then $\sigma \in \heads$.

\end{proof}

To see $\gvf$ satisfies \partinref{eraw}{gvf} we have the following lemma:

\begin{lemma}[Consistency]\label{lem:assignment}
    The tuple $\gvf$ is a gradient vector
    field consistent with $f_0$.
\end{lemma}

\begin{proof}
    Let $\eps>0$ and $d=\dim{K}$.
    Let $\gvf =$ \callerchild.
    We define
    $$\delta :=   \min \{\eps, \min_{v,w \in K_0} |f(v) - f(w)| \}.$$
    We define $f \colon K \to \R$ recursively as follows: for all
    vertices $v \in K_0$, define $f(v) := f_0(v)$.
    Now,
    assume that~$f$ is defined on the $i$-simplices, for some $i \geq 0$.
    For each~$\sigma \in K_{i+1}$, we initially assign~$f(\sigma)=   \max_{\tau \prec
    \sigma} f(\tau)$, then we update:
    \begin{equation}\label{eq:addlrmatch}
        f(\sigma) = f(\sigma) +
        \begin{cases}
            - \lrmatchoffset & \text{if $\sigma$ is a left-right parent;} \\
            \lrmatchoffset & \text{otherwise,}
        \end{cases}
    \end{equation}
    where~$j$ is the index
    of~$\sigma$ in the lexicographic ordering of all simplices.
    We make one final update:
    \begin{equation}\label{eq:addlrchild}
        f(\sigma) = f(\sigma) +
        \begin{cases}
            \lrchildoffset & \text{if $\sigma$ is a left-right child;} \\
            0 & \text{otherwise.}
        \end{cases}
    \end{equation}
    We need to show that~$\gvf$ and $f$ satisfy the three properties in
    \defref{consistent}.

    First, we show \partdefref{consistent}{gvf} holds for $f$ as defined above
    (that~$\gvf$ is the GVF corresponding to~$f$).
    Let~$\othergvf$ be the GVF corresponding to $f$. Since $\heads$, $\tails$,
    $\criticals$ partitions $K$ by \lemref{partition}, it suffices to show that
    $\matches$ is a bijection and $\matches=\matches'$. The only time that simplices are added
    to $\heads$ or $\tails$ happens directly alongside when pairs are added to $\matches$ in
    lines~\ref{line:erchild-addh} and~\ref{line:erchild-defro}, forcing
    that~$\matches$ must be a match.

    Let $(\tau,\sigma) \in \matches$.
    Let $i = \dim{\sigma}$.
    By \lemref{wefindtails}, $\sigma$ is a left-right parent and $\tau =
    \rchild{\sigma}$, which means that $(\tau,\sigma)$ is a left-right pair.
    We follow the computation of~$f(\sigma)$.
    Since $(\tau, \sigma)$ is a left-right pair,
    $\tau$ is the rightmost face of $\sigma$,
    which means $f(\sigma)$ is initialized to $f(\tau)$.
    Since $\sigma$ is a left-right parent, $f(\sigma)$ is
    updated by~\eqref{addlrmatch} to
    $f(\sigma) = f(\sigma) - \lrmatchoffset$. Since~$\sigma$ is not a left-right
    child, nothing changes in \eqref{addlrchild}. Thus,~$f(\sigma) < f(\tau)$.
    Next, let $\tau' \prec \sigma$ such that $\tau' \neq \tau$ and
    $\dim{\tau'}=i-1$. We follow the
    computation of $f(\tau')$.
    Since~$\tau$ is the only face of $\sigma$ that is a left-right
    child, for any other $\tau' \prec \sigma$,
    \eqref{addlrchild}, adds zero to the definition of $f(\tau')$.
    Recalling that \eqref{addlrchild} adds $\lrchildoffsetlowerdim$ to the definition
    of~$f(\tau)$, we find that
    $f(\tau) \geq f(\tau') + \lrchildoffsetlowerdim$, and
    $$
        f(\sigma)= f(\tau) - \lrmatchoffset
            \geq f(\tau') + \lrchildoffsetlowerdim - \lrmatchoffset \geq
            f(\tau').
    $$
    Because $\sigma$ may be any arbitrary left-right parent, we can 
    guarantee that the above inequality is valid for any $(\sigma,\tau)
    \in \matches$ when related to any other faces of $\sigma$. Thus, $f$
    is discrete Morse, since it is impossible for $f$ to violate the 
    inequality given in \defref{morsedefn}. 
    
    Since $f(\tau) > f(\sigma)$ and $f$ is a discrete Morse function,
    we obtain~$(\tau,\sigma) \in \matches'$.  Each of these statements are biconditional, so
    we have shown that $\matches = \matches'$.

    \partdefref{consistent}{restr} ($f|_{K_0}=f_0$) holds trivially.

    Finally, we show \partdefref{consistent}{eps} holds
    (that $|f(\sigma) - \max_{v \in \sigma}f_0(v)| \leq \eps$).
    By construction,
    $$
        |f(\sigma - \max_{v \in \sigma}f_0(v)|
            \leq \left( \sum_{i=1}^d 2^{-i} \right) \delta
            = (1-2^{-d})\delta
            < \eps.
    $$
\end{proof}

Properties \partinref{eraw}{gvf}, \partinref{eraw}{ifftails}, and
   \partinref{eraw}{composition} are quite restrictive. In fact, they
   uniquely determine a GVF, as we now show.

\begin{theorem}[Unique GVF]\label{thm:unique}
    Let~$K$ be a simplicial complex and let $f_0 \colon K_0 \to \R$ be an
    injective function.
    There is exactly one gradient vector field, $\gvf$, with the following
    two~properties:
    \begin{enumerate}[(i)]
        \item $\gvf$ is consistent with~$f_0$.\label{part:unique-gvf}
        \item For all $\sigma \in K$,  $\sigma \in \heads$ if and only if
            $\sigma$ is a left-right~parent.\label{part:unique-ifftails}
        \item For all $\sigma \in \heads$,
            $\matches(\rchild{\sigma})=\sigma$.\label{part:unique-define}
    \end{enumerate}
\end{theorem}

\begin{proof}
    Let $K$ and $f_0$ be as defined in the theorem statement.
    Let $\maxh \colon K \to \R$ be defined for each simplex~$\sigma \in K$
    by~$\maxh(\sigma) := \max_{v \in \sigma} f_0(v)$.
    Let~$\gvf$ and
    $\othergvf$ be two GVFs that satisfy~\partinref{unique}{gvf},
    \partinref{unique}{ifftails}, and~\partinref{unique}{define}.

    Let $\sigma \in \heads$. By the forward direction of
    \partinref{unique}{ifftails}, we know that $\sigma$ is a left-right
    parent.  By the backward direction of
    \partinref{unique}{ifftails}, we know that $\sigma \in H'$.
    Thus, we have shown that $\heads \subseteq \heads'$. Repeating this argument
    by swapping the roles of $\heads$ and $\heads'$ gives us $\heads' \subseteq \heads$.

    Since $\sigma \in H = H'$ and because \partinref{unique}{define} holds, 
    we have shown that~$\sigma$
    is paired with $\rchild{\sigma}$ in both matchings, and specifically 
    ~$\matches(\rchild{\sigma})=\sigma = \matches'(\rchild{\sigma})$.
    Since $\matches$ and $\matches'$ are
    bijections by~\partinref{unique}{gvf}, we also know
    that:
    $$\tails = \{ \tau \in K ~|~ \exists \sigma \in \heads \text{ s.t.\ }
    \matches(\rchild{\sigma})=\sigma \}=\tails'.$$
    Thus, $\tails=\tails'$ and~$\matches=\matches'$.

    Finally, we conclude:
    $$\criticals = K \setminus (\tails \cup \heads)
        = K \setminus (\tails' \cup \heads')
        = \criticals',$$
    which means that $\gvf$ and $\othergvf$ are the same GVF.
    Thus, we conclude that the gradient vector field satisfying
    \partinref{unique}{gvf}, \partinref{unique}{ifftails}, and
    \partinref{unique}{define} is unique.
\end{proof}

Since \erchild\ and \eraw\ both satisfy the hypothesis of \thmref{unique},
the outputs of the algorithms must be the same.

\begin{theorem}[Algorithm Equivalence]\label{thm:tada}
    Let~$K$ be a simplicial complex and let $f_0 \colon K_0 \to \R$ be an
    injective function. Then
    \eraw($K$, $f_0$) and \erchild($K$, $f_0$) yield identical outputs.
\end{theorem}

\begin{proof}
    By \thmref{eraw} and \lemref{kingfindmatch}, the output of \eraw\ satisfies the properties in \thmref{unique}.
    By \lemref{assignment}, \lemref{wefindmatch},  and \lemref{wefindtails},
    the output of \erchild\ satisfies the properties of \thmref{unique}.
       Then, by \thmref{unique}, \eraw\ and \erchild\  are equivalent.

\end{proof}

When we consider the runtime and space usage of \erchild, we find
the following:

\begin{theorem}[New Runtime]\label{thm:our-runtime}
    Given a simplicial complex $K$ (represented as a Hasse diagram),
    and an injective function $f_0: K_0 \rightarrow \mathbb{R}$,
    $\erchild$ computes a GVF consistent with $f_0$ in \ercrun{} time
    and uses \ercspace{} space.
\end{theorem}

\begin{proof}
First, line \shortlineref{decorate}{erchild} decorates the Hasse diagram.
By \lemref{decoration}, the decoration takes $O(dn)$ time and $O(n)$ space.
\shortlinerangeref{downd}{downdend}{erchild}, process each node
of the decorated Hasse diagram.
Each iteration of the loop is $O(1)$ in time and space because
all required data was computed while decorating.
As there are $n-n_0$ nodes to process,
\shortlinerangeref{downd}{downdend}{erchild}
takes $O(n)$ time and uses $O(1)$ space.
Finally, we iterate over the zero-simplices in $O(n_0)$ time.

The bottleneck of space and time usage of the algorithm is decorating the
Hasse diagram, therefore,
the algorithm takes $O(dn)$ time and $O(n)$ space.
\end{proof}

%% file: body/hasse-dgm.tex
We assume that KKM~\cite{king} represent $K$ in a
\emph{standard Hasse diagram data structure} $\hasseK$,
which can be encoded as an adjacency list representation for a graph.
Each simplex $\sigma \in K$ is represented by a node in $\hasseK$.
We abuse notation and write $\sigma \in \hasseK$ as the corresponding node.
Two simplices $\sigma,\tau \in \hasseK$ are connected
by an edge from $\sigma$ to $\tau$ if
$\sigma$ is a codimension one face of $\tau$.
For a node $\sigma \in \hasseK$, we partition its edges into two sets,
$\up{\sigma}$ and $\down{\sigma}$ as the edges in which $\sigma$ is a
face or coface, respectively.

For $p \in \N$, we denote the nodes of $\hasseK$ corresponding to the
$p$-simplices of $K$
as $\hasseKI{p}$ and we store each $\hasseKI{p}$ in its own
set that can be accessed in $O(1)$ time.
Note that there is no requirement about the ordering of
the edges  or the nodes in each $\hasseKI{p}$.
See \figref{cat} for an example of the data structure.

\begin{figure}[ht!]
    \begin{subfigure}[b]{0.3\columnwidth}
        \includegraphics[height=.9in]{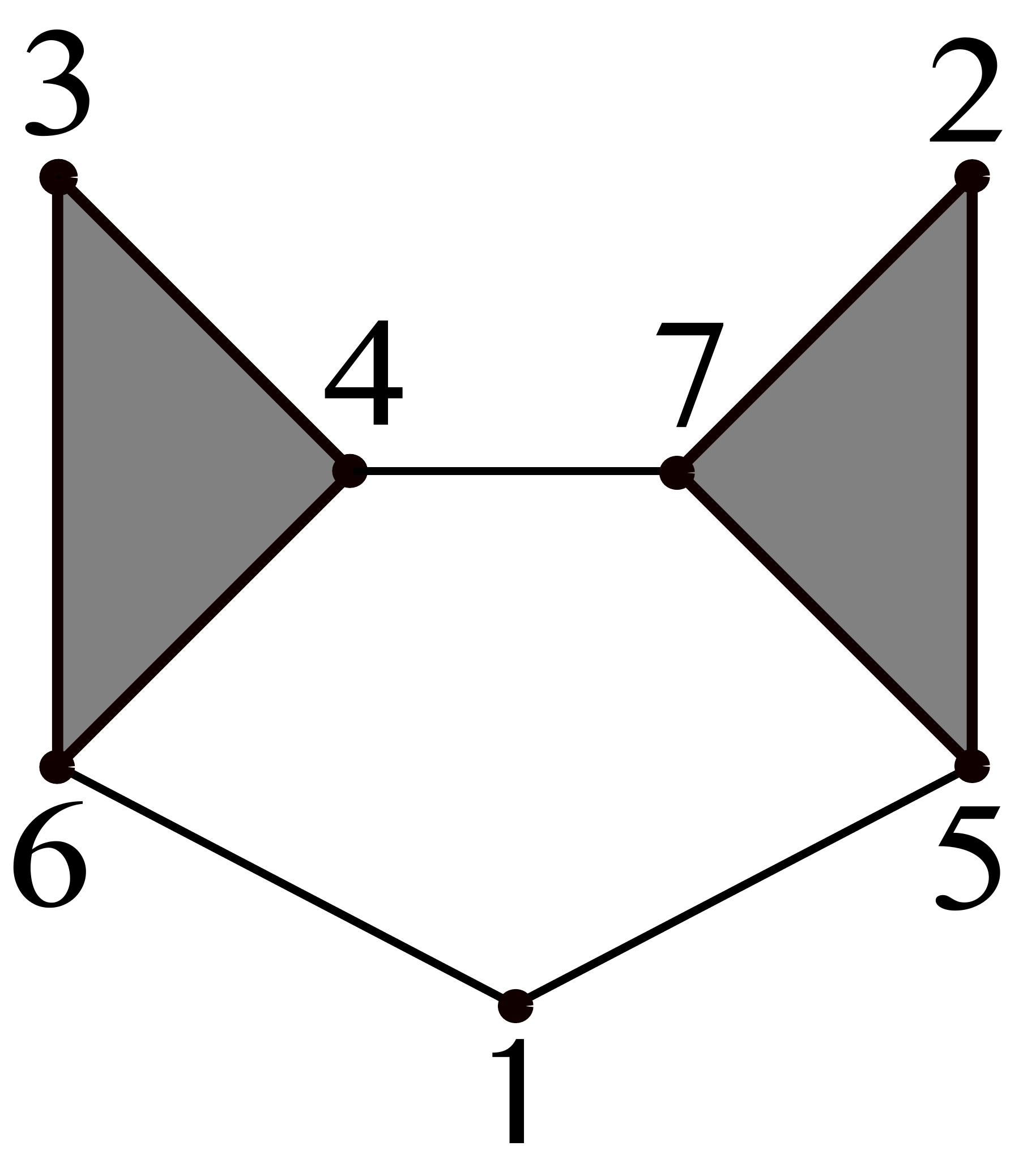}
    \end{subfigure}
    \begin{subfigure}[b]{0.6\columnwidth}
        \includegraphics[width=2.25in]{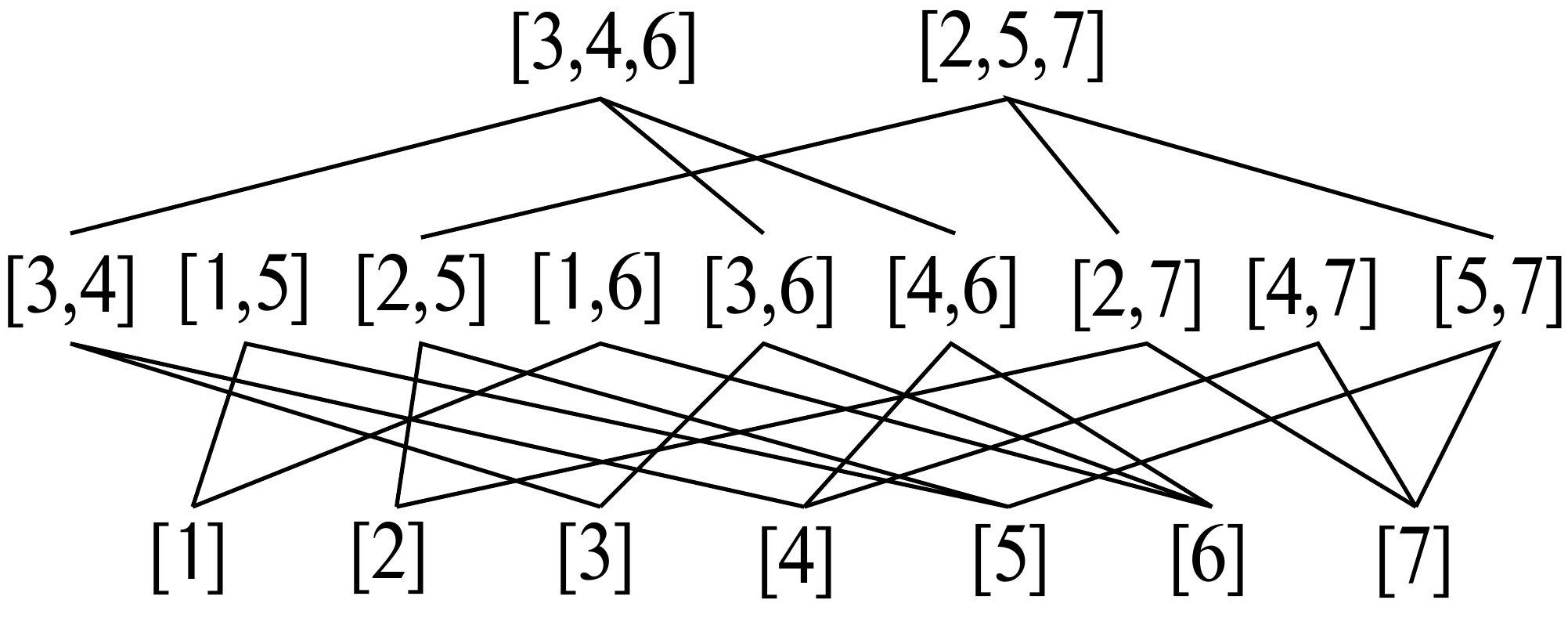}
    \end{subfigure}
    \caption{
        {\it Left}: A simplicial complex with function values assigned to
        the vertices.
        {\it Right}: The Hasse diagram of the simplicial complex.
    }
    \label{fig:cat}\label{fig:hasse}
\end{figure}

For our algorithm, we decorate each node of $\hasseK$ with additional data.
For clarity, we denote the decorated data data structure as $\dhasseK$.
Next, we describe the additional data stored in each node and how to
initialize the data.
Consider $\sigma \in \dhasseKI{p}$ and
define $\maxh(\sigma) := \max_{v \in \sigma} f_0(v)$.
Each node stores $\maxh(\sigma)$,
the rightmost child $\rchild{\sigma}$ and
leftmost parent $\lparent{\sigma}$.

Next, we describe how to initialize the data and summarize with the following
lemma.

\begin{lemma}[Hasse decoration]\label{lem:decoration}
    Given a simplicial complex $K$ with $n$ simplices and $\dim{K} = d$.
    The decorated Hasse diagram uses $O(n)$ additional space.
    We can decorate the Hasse digram $K$ in $O(dn)$ time.
\end{lemma}

\begin{proof}

We begin by analyzing the space complexity.
For each node, we store a constant amount of additional data.
Thus, the decorated Hasse diagram uses $O(n)$ additional space.

Next, we analyze the time complexity.
To decorate $\hasseK$ for each node $\sigma \in \hasseK$,
we must compute $\maxh(\sigma)$, $\rchild{\sigma}$, and $\lparent{\sigma}$.
Let $p = \dim{\sigma}$. We proceed in three steps.

First we compute $\maxh$.
In general, computing $\maxh(\sigma)$ takes $O(p)$ time, since there may be
no more than $p$ vertices which compose any $\sigma$.
Let $\tau_1$ and $\tau_2$ be distinct codimension one faces of $\sigma$.
Observe that $\maxh(\sigma) = \max(\maxh(\tau_1), \maxh(\tau_2))$.
Thus, if we know the function values for
$\hasseKI{p-1}$, we can compute and store all function values of all
nodes in $\hasseK$ in $\Theta(n)$ time.

Second we compute $\rchild{\sigma}$ by brute force.
We iterate over all edges in $\down{\sigma}$ to find its
largest face under lexicographic ordering.
Since a $p$-simplex $\sigma$ has $p+1$ down edges, computing $\rchild{\sigma}$
for $\sigma$ takes $O(p)$ time.
As $\dim{K} = d$, and there are $n$ nodes,
we can then compute $\rchildName$ for all nodes in $O(dn)$ time.

Third, we compute $\lparentName$, also by brute force.
We iterate over all edges in $\up{\sigma}$ to find its
smallest lexicographical coface.
While we cannot bound $\up{\sigma}$ as easily as $\down{\sigma}$, we do know
that when computing $\lparentName$ we can charge each edge in the Hasse diagram 
for one comparison. Observe that when computing $\rchildName$, we can similarly
charge each comparison to an edge. Then, from computing $\rchildName$, we know
the total number of comparisons is $O(dn)$.
Thus, the total number of comparisons for computing $\lparentName$
is also $O(dn)$.

As each step takes $O(dn)$ time, decorating $\hasseK$ takes $O(dn)$ time.
\end{proof}

%% file: body/algo-right-child.tex
\begin{algorithm}
    \caption{\erchild}\label{alg:erchild}
    \begin{algorithmic}[1]

        \Require simp.\ complx.\ $K$, injective fcn.\ $f_0: K_0\rightarrow \mathbb{R}$
        \Ensure a GVF consistent with $f_0$

        \State{$\dhasseK \gets $ decorate the Hasse diagram of $K$} \Comment{\shortlemref{decoration}|}
        \label{line:erchild-decorate}

        \State $\tails \gets \emptyset$, $\heads \gets \emptyset$, $\criticals
        \gets \emptyset$, $\matches \gets \emptyset$

        \For {$i = \dim{K}$ to $1$}
        \label{line:erchild-downd}

        \For{$\sigma \in \dhasseKI{i}$}\label{line:erchild-ltor}

        \If{$\sigma$ is assigned}
        \State {\bf continue}
        \EndIf

        \If{$\sigma$ is a left-right parent}\label{line:erchild-leftparent}

        \State Add $\rchild{\sigma}$ to $\tails$; Add $\sigma$ to $\heads$\label{line:erchild-addh}

        \State{Add $(\rchild{\sigma},\sigma)$ to $\matches$}
        \label{line:erchild-defro}

        \State Mark $\sigma$ and $\rchild{\sigma}$ as assigned\label{line:erchild-abass}

        \Else

        \State Add $\sigma$ to $C$\label{line:erchild-addc}
        \State Mark $\sigma$ as assigned\label{line:erchild-cass}

        \EndIf

        \EndFor
        \EndFor
        \label{line:erchild-downdend}

        \State{Add any unassigned zero-simplices to $\criticals$}
        \label{line:erchild-leftovers}\\

        \Return $(\tails, \heads, \criticals, r)$
        \label{line:erchild-return}

    \end{algorithmic}
\end{algorithm}

Next, we describe the main algorithm.
Given a simplicial complex $K$ (represented as a Hasse diagram),
and an injective function $f_0: K_0 \rightarrow \mathbb{R}$,
$\erchild$ computes a GVF consistent with $f_0$.

\algref{erchild} has three main steps.
First, we create a decorated Hasse diagram.
Second, we process each level of the Hasse diagram from top to bottom.
For each unassigned simplex, we check for a left-right parent node,
and use the results to build up a GVF.
Third, we process unassigned zero-simplices. 
See \figref{hasses} for an example.

\begin{figure}[!htb]
        \centering
        \begin{subfigure}[b]{0.35\textwidth}\label{fig:alg1}
        \includegraphics[width=\textwidth]{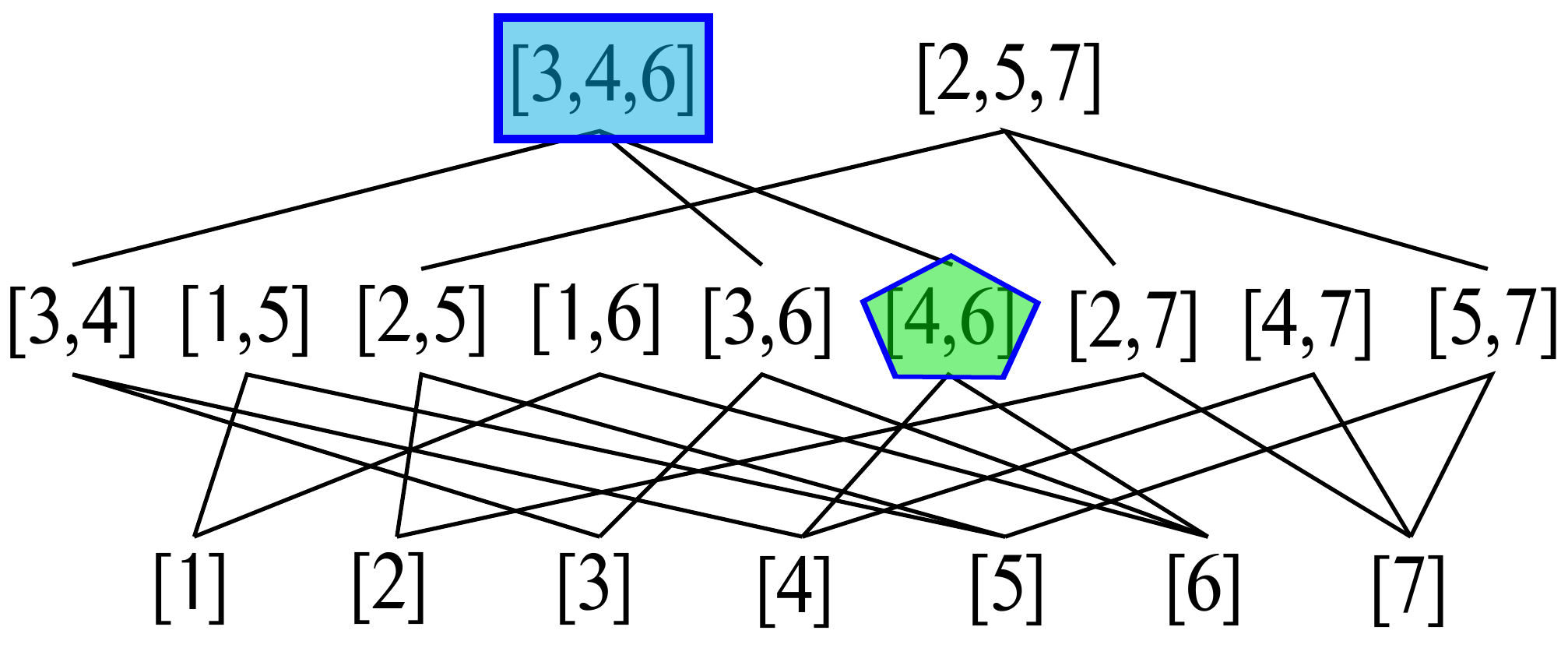}
        \caption{The simplex $[3,4,6]$ is a left-right parent
        because $\lparentName \circ \rchild{[3,4,6]} = \lparentName (4,6)=[3,4,6]$.
        The algorithm
         adds $[3,4,6]$ to $\heads$ and $[4,6]$ to $\tails$.}
	 \end{subfigure}

        \begin{subfigure}[b]{0.35\textwidth}\label{fig:alg2}
        \includegraphics[width=\textwidth]{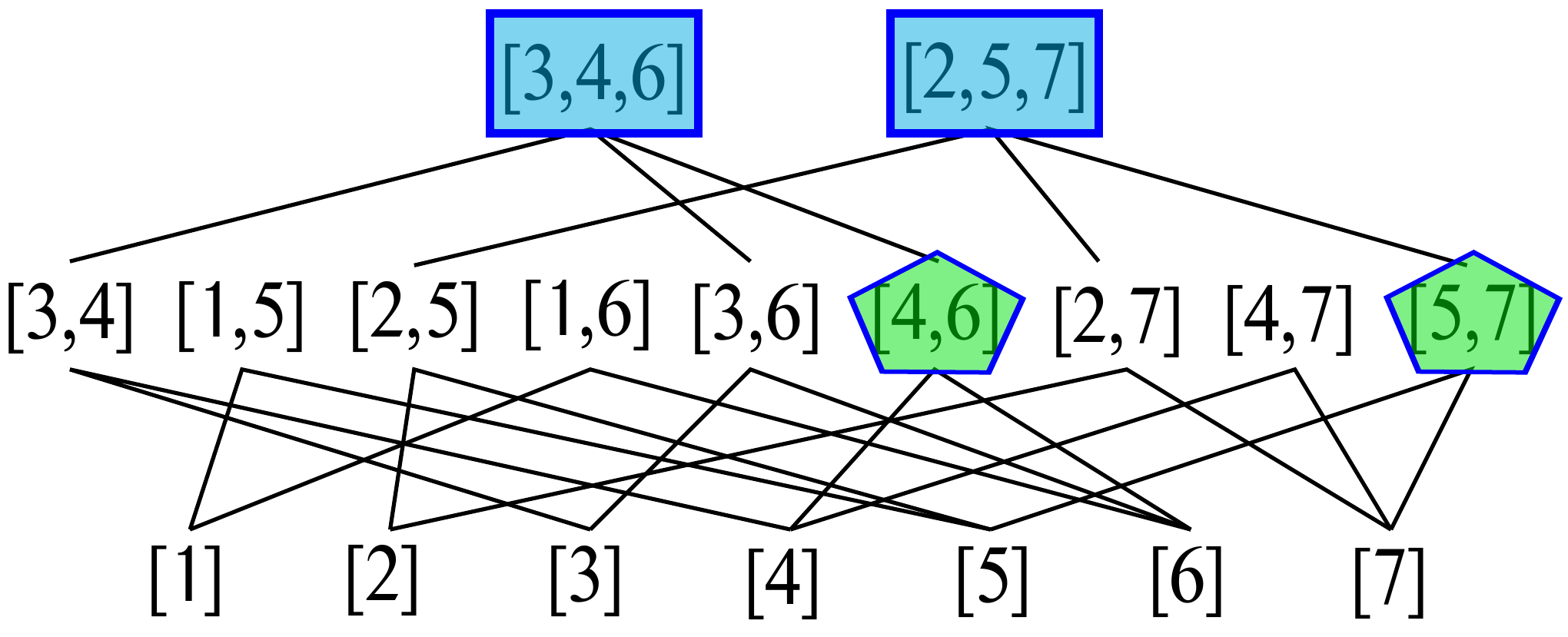}
        \caption{The simplex $[2,5,7]$ is also a left-right parent.
        The algorithm adds $[2,5,7]$ to $\heads$ and $[5,7]$ to \tails.}
        \end{subfigure}

             \begin{subfigure}[b]{0.35\textwidth}\label{fig:alg3}
        \includegraphics[width=\textwidth]{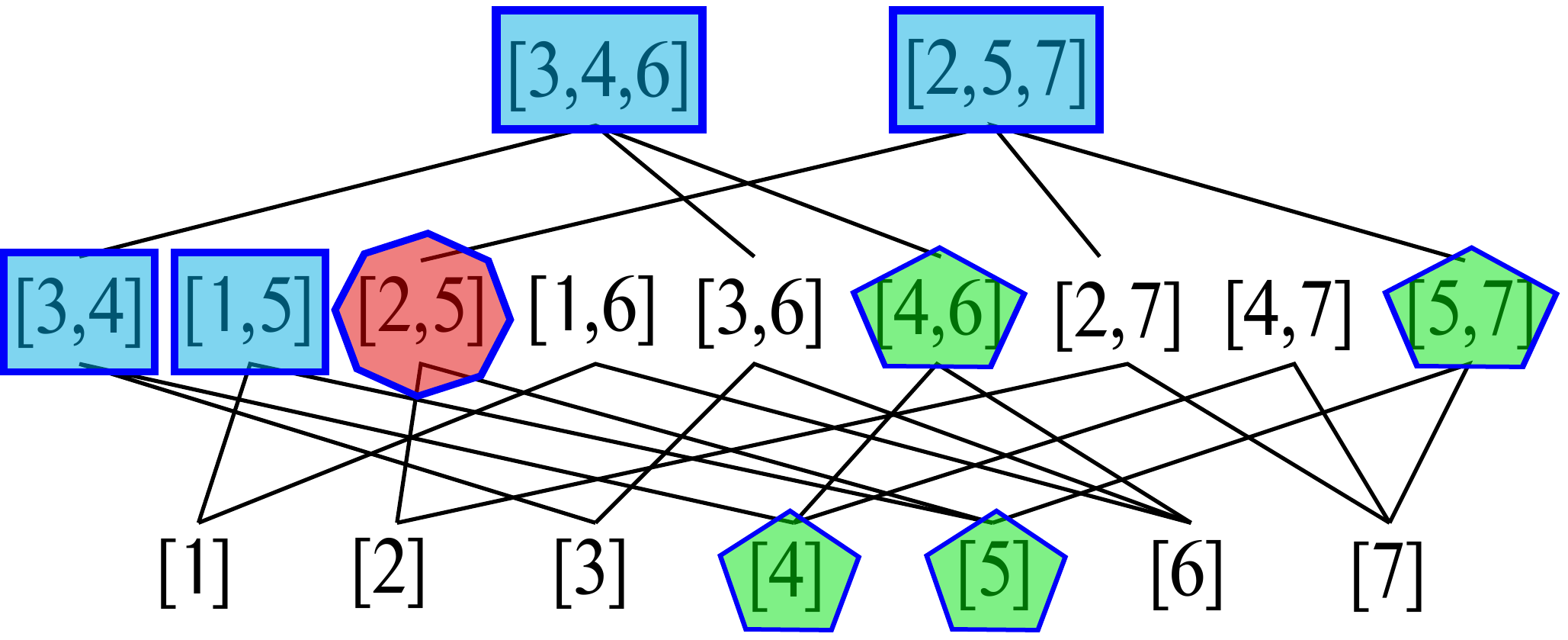}
        \caption{Both $[3,4]$ and $[1,5]$ are left-right parents. 
        The simplex $[2,5]$ is not a left-right parent because
        $\lparentName \circ \rchild{[2,5]} = [1,5] \neq[2,5]$. The algorithm
        adds $[2,5]$ to \criticals.}
        \end{subfigure}
   
               \begin{subfigure}[b]{0.35\textwidth}\label{fig:alg4}
        \includegraphics[width=\textwidth]{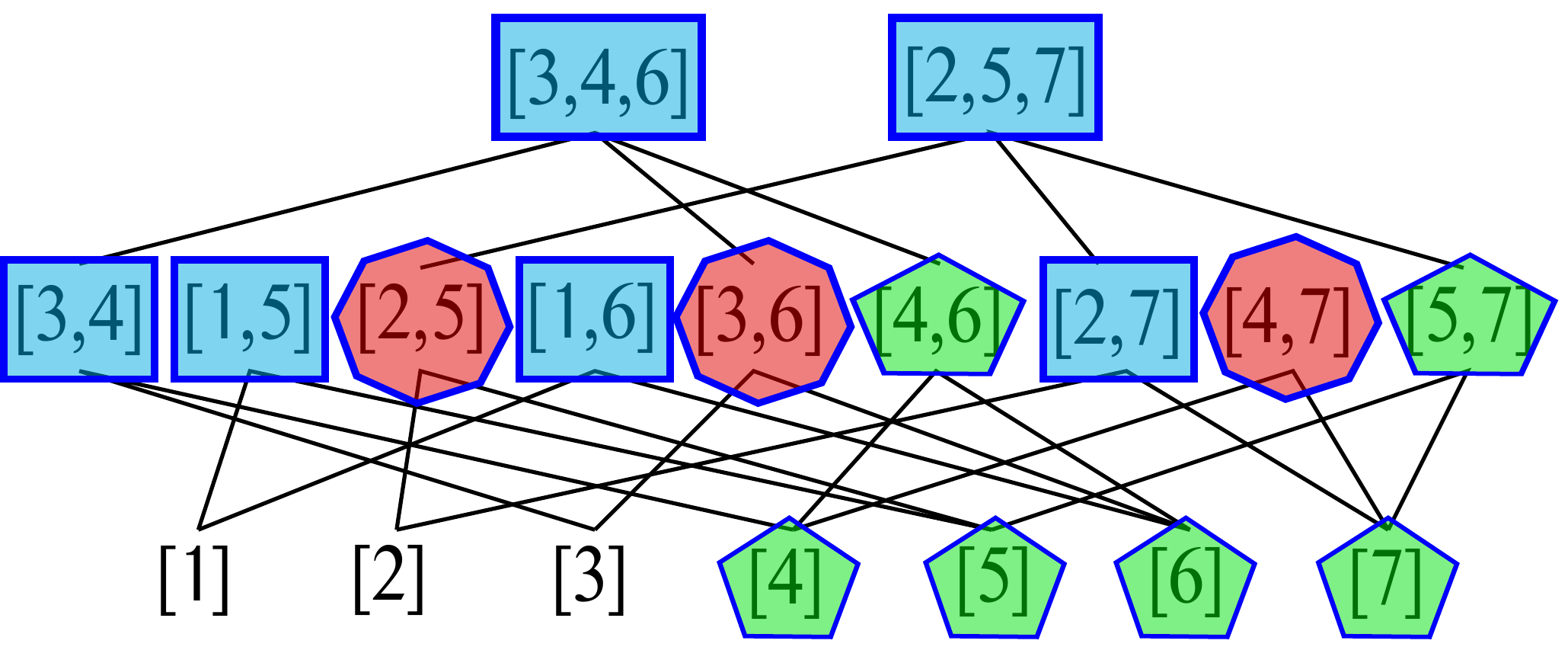}
        \caption{The loop in \shortlineref{downd}{erchild} is complete.}
        \end{subfigure}
        
               \begin{subfigure}[b]{0.35\textwidth}\label{fig:alg5}
        \includegraphics[width=\textwidth]{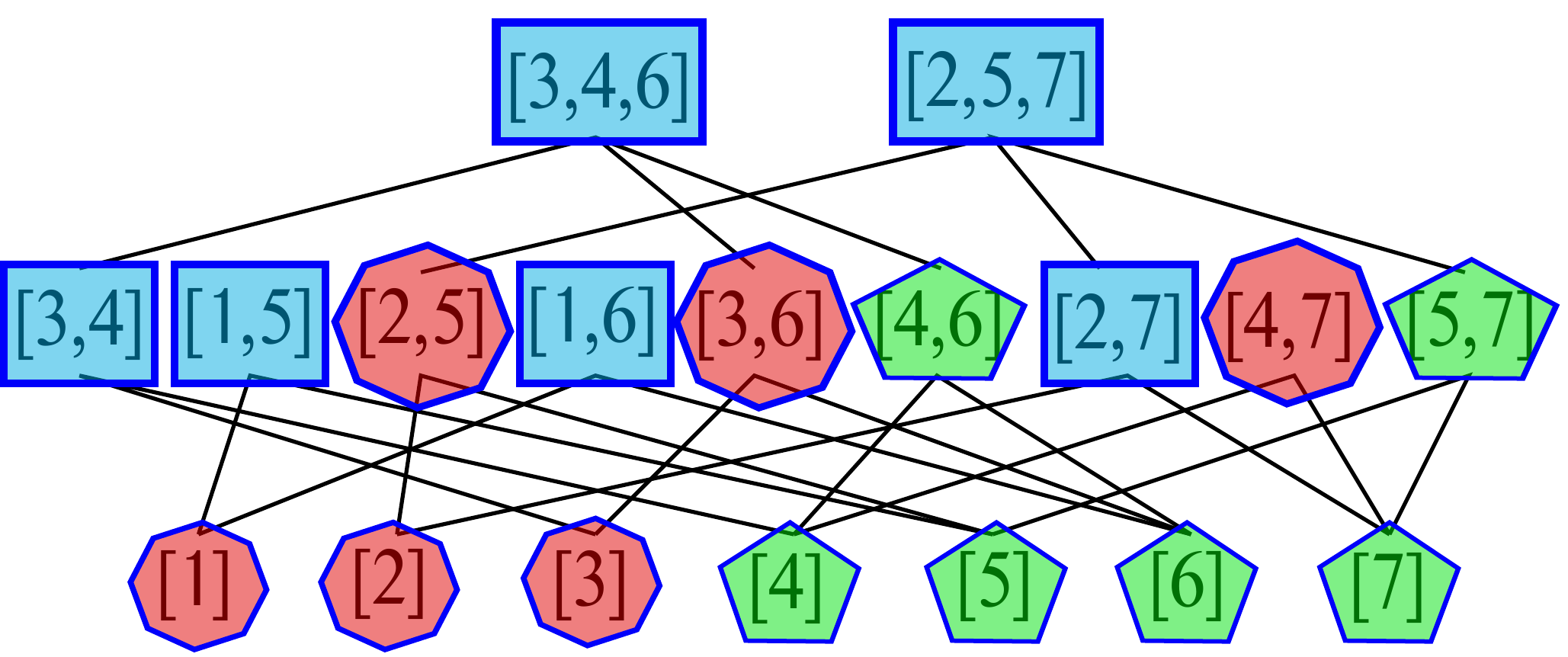}
        \caption{The algorithm adds unassigned vertices to $\criticals$.}

        \end{subfigure}

		\caption{Here we see a visualization of \algref{erchild}, \erchild\, on the complex shown in \figref{cat}. 
		\algref{erchild} partitions the nodes of the Hasse diagram into three sets, $\heads, \tails$ and $\criticals$.
		Elements of $\heads$ are represented by blue rectangles, elements of $\tails$ by green pentagons
		and elements of $\criticals$ by red hexagons.}
			\label{fig:hasses}
\end{figure}

%% file: body/discussion.tex
\section{Discussion}\label{discussion}

In this paper, we identified properties of the \extract{} and
\eraw{} algorithms~\cite{king}.
We used these properties to simplify \eraw{} to the 
equivalent algorithm \erchild{}.
Our simplification improves the 
runtime from \erawrun{} to \ercrun{}.

There are several possible extensions of this work. The
problem of finding tight bounds on the runtime of \extract{} is interesting
and open. We plan
to implement our approach on high dimensional data sets, and to
further improve to the runtime.
We intend to explore a cancellation algorithm that performs the same task as
\ecancel, eliminating critical pairs with small persistence.
Our conjectured cancellation algorithm iterates over critical simplices and
applies \erchild.

Constructing Morse functions that do not require preassigned
function values on the vertices is a related area of active research.
The problem of finding a Morse function with a minimum number of
critical simplices is NP-hard \cite{joswig04}. 
In \cite{bauer-rathod-18}, Bauer and Rathod show that 
for a simplicial complex of dimension $d \geq 3$ with $n$ simplices, 
it is NP-hard to approximate a Morse matching with
a minimum number of critical simplices within a factor of 
$O(n^{1-\epsilon})$, for any $\epsilon> 0$.
The question is open for 2-dimensional simplicial
complexes.

%% file: body/acknowledgements.tex
This material is based upon work supported by the National Science Foundation
under the following grants:
CCF 1618605 \& DMS 1854336 (BTF) and
DBI 1661530 (DLM).
Additionally, BH thanks the Montana State Undergraduate Scholars Program.
All authors thank Nick Scoville for introducing us to KKM~\cite{king}
and for his thoughtful discussions.

%% file: body/append-extract.tex
To put our result in context, we now provide a glimpse into the inner workings of \extract, and reveal the underlying properties of \eraw\ which give it an identical output to \erchild. We also provide a formal runtime analysis of \eraw\ to verify that \erchild\ provides an improved time complexity.

\subsection{Subroutines for \extract}\label{append:kingsubroutines}
\input{body/append-extract-algos}

\subsection{Analysis of \eraw}\label{append:king}

\input{body/append-extract-king}

%% file: body/append-extract-algos.tex
In this section, we recall the algorithms proposed by KKM~\cite{king}.
Note that we made some slight modifications to the presentation of KKM's
initial description to improve readability.
The modifications do not affect the asymptotic time or space
used by the algorithm,
although it does remove some redundant computation.

In particular, we modified the inputs to explicitly pass around
a GVF so that the inputs of each algorithm are clear.
We simplified notation and inlined the subroutine $\cancel$.
From the previous modifications,
we observed that the algorithm recomputes a gradient path
that is currently in scope
and so we simply unpack the path on \lineref{unpackGamma}{ecancel}.

\eraw\ computes the lower link of each vertex $v$ in a \simp, and assigns $v \in \criticals$
if $\llinkK{v} = \emptyset$. If $\llinkK{v} \not = \emptyset$, its lower link is
recursively inputted into \eraw\ and this recursion continues until an empty lower link is reached. 
When the lower link is not empty, \eraw\ assigns $v \in T$ and the smallest function valued vertex 
$\omega_0$ in $\llinkK{v}$ is combined with $v$ and added to $\heads$, carrying with this assignment 
a mapping $m$ from $\omega_0*v \in \heads$
to $v \in \tails$. As the recursion continues, higher dimensional simplices in the lower start of $v$
are able to be assigned to both $\heads$
and $\tails$ based on combinations consistent with the assignments of the vertices and the original 
mappings of $\matches$.
Higher dimensional critical cells are assigned similarly by combining the current vertex and each 
previously computed $\sigma \in \criticals$ from the last recursion, until all simplices have been assigned.

Then, because \eraw\ may have extraneous critical cells,
\cancel\ works to reduce the number of critical cells by locating ``redundant" gradient paths to a
critical simplex and reversing them after the first pass by \eraw, refining the output of \extract\ .

\begin{algorithm}
    \caption{\cite{king} \eraw}\label{alg:eraw}
    \begin{algorithmic}[1]
        \Require simp.\ complx.\ $K$, injective fcn.\ $f_0: K_0\rightarrow \mathbb{R}$
        \Ensure a GVF consistent with $f_0$
        \State $\tails \gets \emptyset$, $\heads \gets \emptyset$, $\criticals
            \gets \emptyset$, $\matches \gets \emptyset$
        \ForAll{$v\in K$}\label{line:eraw-selectv}
            \State Let $K':=$ the lower link of $v$. \label{line:eraw-llcall}
            \If{$K'=\emptyset$}
                \State Add $v$ to $\criticals$\label{line:eraw-addvc}
            \Else
            \State Add $v$ to $\tails$\label{line:eraw-addvt}
            \State $(\tails',\heads',\criticals',\matches')\leftarrow
                $\extract$(K',f_0,\infty)$\label{line:eraw-callexact}
            \State $w_0 \gets \argmin_{w \in \criticals'_0} \{ f_0(w) \}$\label{line:eraw-omnot}
            \State Add $w_0v$ to $\heads$\label{line:eraw-addh1}
            \State Define $\matches(w_0*v):=v$
            \State For each $\sigma\in \criticals' \setminus \{w_0\}$, add
                $v*\sigma$ to $\criticals$\label{line:eraw-highc}
                \ForAll{$\sigma\in \tails'$}
                    \State Add $v*\sigma$ to $\tails$\label{line:eraw-addt}
                    \State Add $v*\matches'(\sigma)$ to $\heads$\label{line:eraw-addh2}
                    \State Define $\matches(v*\sigma)=v*\matches'(\sigma)$
                \EndFor
            \EndIf
        \EndFor\\
        \Return $\gvf$
    \end{algorithmic}
\end{algorithm}

\begin{algorithm}
    \caption{\cite{king} \ecancel}\label{alg:ecancel}
    \begin{algorithmic}[1]
        \Require \simp\ $K$,
            injective function $f_0\colon K_0 \to \R$,
            $p\geq 0$, $j \in \N$, and GVF $\gamma$
        \Ensure Gradient vector field on $K$

        \State Let $\gvf$ be the four components of~$\gamma$
        \ForAll{$\sigma \in \criticals_j$}
            \State $s \gets \max_{v \in \sigma} f_0(v)$
            \State $S \gets \{ \Gamma ~|~ \Gamma \in \gradpaths ,
                s - \max_{w \in \Gamma_L} f_0(w) < p  \}$
            \ForAll{$\Gamma \in S$}
                    \State $m_{\Gamma} \gets \infty$
                \If{$\Gamma_L \neq \Gamma'_L$ for any other $\Gamma' \in S$}
                    \State $m_{\Gamma} \gets \max_{w \in \Gamma_L}{f_0(w)}$
                \EndIf
            \EndFor
            \State $\Gamma^* \gets \argmin_{ \Gamma \in S}
                \{ m_{\Gamma}\}$
            \If{$m_{\Gamma^*} \neq \infty$}

                \State $\{\sigma_1, \tau_1, \cdots, \sigma_k, \tau_k\} \gets \Gamma^*$
                \label{line:ecancel-unpackGamma}

                \State Remove $\tau_k, \sigma_1$  from $\criticals$
                \State Add $\tau_k$ to $\tails$; Add $\sigma_1$ to $\heads$
                \State Add $(\tau_k, \sigma_{k})$ to $\matches$
                \For{$i=1,...,k-1$}
                    \State Remove $(\tau_i, \sigma_{i+1})$ from $\matches$
                    \State Add $(\tau_i,\sigma_i)$ to $\matches$
                \EndFor
            \EndIf
        \EndFor\\
        \Return $\gvf$
    \end{algorithmic}
\end{algorithm}

Let $p \in \N$, $\sigma \in K_p$ be a critical simplex.
Let $\gradpaths$ denote the set of all nontrival gradient
path starting at $\sigma \in \criticals_j$ and ending in $\criticals_{j-1}$.

%% file: body/append-extract-king.tex
In this appendix, we provide the analysis \algref{extract} from \secref{king}.
In what follows, let $K$ be a simplicial complex and let $f_0 \colon K_0 \to \R$
be an injective function.

\begin{lemma}[Raw Heads are Parents]\label{lem:kingfindmatch}
    Let $\gvf$ be the output of $\eraw(K,f_0)$.
    Every simplex in $\heads$ is a left-right parent.
    Furthermore, for all $\sigma \in \heads$,
    $\matches(\rchild{\sigma})=\sigma$.
\end{lemma}

\begin{proof}

    Let $\sigma \in \heads$.  We show that $\sigma$ is a left-right
    parent by induction on the dimension of $K$.
    When $\dim{K}=1,$
    $\sigma$ is an edge, and $\sigma=\matches(\tau)$ for some vertex $\tau\in T$.
    In \lineref{addh1}{eraw} $\sigma$
    is defined as
    $\matches(\tau)=[w_0, \tau]$ where $w_0\in C'_0$ so that $f_0(w_0)$ is smallest.
    So, $\sigma$ is a left-right parent. Furthermore,
    $\matches(\rchild{\sigma})=\matches(\rchild{[w_0,\tau]})=\matches(\tau)=\sigma$.

    Suppose every $\sigma \in H$ is a left-right parent when $\dim{K}\leq d$ and consider $\dim{K}=d+1$.
    If $\sigma$ is a $(d+1)-$simplex, $\sigma=\matches(\tau)$ is defined in
     \lineref{addh2}{eraw},
    when a vertex $v$ is selected in \lineref{selectv}{eraw}. We extend
    the GVF on the $\llinkK{v}$ to include the lower star of $v$.
    We have $\matches(\tau)=v*\matches'(\alpha)$ where $\alpha=[v_1,\ldots, v_d]\in T',$
    $\matches'(\alpha)=[v_0,v_1\ldots v_d]\in H'$.
    Since $\alpha$ and $\matches'(\alpha)$ are in the $\llinkK{v}$ we have
    $f(v_i)<f(v)$ for $0\leq i \leq d$. Then $\tau=v*\alpha=[v_1,\ldots,v_d,v]$
    and $\sigma=\matches(\tau)=[v_0,v_1,\ldots,v_d,v]$.

    By the induction hypothesis $\matches'(\alpha)\in H'$ is a left-right parent.
    If $\sigma$ is not a left-right parent,
    we can remove $v$ from $\sigma$ and $\tau$
    and contradict that $\matches'(\alpha)\in H'$.

    Furthermore, $\matches(\rchild{\sigma})=\matches(\rchild{\matches(\tau)})=
    \matches(\rchild{[v_0,v_1,\ldots,v_d,v]})=\matches([v_1,\ldots,v_d,v])=\matches(\tau)=\sigma$.
    This proves the claim.
\end{proof}

\begin{lemma}[Raw Parents are Heads]\label{lem:kingfindtails}
    Let $\gvf$ be the output of $\eraw(K,f_0)$.
    Let $\sigma \in K$.  If $\sigma$ is a left-right parent, then $\sigma \in
    \heads$.
\end{lemma}

\begin{proof}
We show if $\sigma \in T \cup C$ then $\sigma$ is not a left-right parent.
First, suppose $\sigma \in T$.
We use induction on $\dim{K}$ to show $\sigma$ is not a left-right parent.
For the base case, $\dim{K}=1, \sigma$ is a vertex and can not be a left-right parent.

Suppose $\sigma \in T$ is not a left-right parent when $\dim{K}\leq d$ and consider $\dim{K}=d+1$.
Then $\sigma$ is added to $T$ in \lineref{addt}{eraw} when a vertex $v$ is selected in \shortlineref{selectv}{eraw}.
As in \lemref{kingfindmatch}, write $\sigma=v*\alpha=[v_1,\ldots,v_d,v]$ for some $\alpha\in T'$.

By the induction hypothesis $\alpha$ is not a left-right parent, thus
$\ell \circ \rchild\alpha\neq \alpha$, and there exists a vertex $v_{-1}$ such that
$\ell ( \rchild \alpha)=[v_{-1},v_2,v_3\ldots,v_d]$ where $f(v_{-1})<f(v_1)$. We have
$\ell \circ \rchild \sigma=\ell \circ \rchild {[v_1,v_2,\ldots,v_d,v]}=
\ell( [v_2,v_3\ldots,v_d,v])=[v_{-1},v_2,\ldots,v_d,v]\neq \sigma$.

Now, suppose $\sigma  \in C.$
There are two places where elements are added to $C,$ \lineref{addvc}{eraw} and
\shortlineref{highc}{eraw}.
In \shortlineref{addvc}{eraw} $c$ is a vertex and can not be a left-right parent.

In \shortlineref{highc}{eraw}
$c$ is defined as $c=v*\alpha$ for some $\alpha=[v_0,v_1,\ldots, v_d]\in C'\backslash w_0$
where $w_0\in C'_0$ so that $f_0(w_0)$ is smallest.
Now $\ell \circ g (c)=\ell \circ g([v_0,v_1,\ldots, v_d,v])=
\ell([v_1,\ldots ,v_d,v])=[w_0,v_1,\ldots,v_d,v]\neq c$.
We have shown that if $\sigma$ is a left-right parent, then
$\sigma\in H$.

\end{proof}

We summarize the properties of \eraw\ in the following theorem.

\kingthm*

\begin{proof}
    \partinref{eraw}{gvf} is proven in Theorem 3.1 of ~\cite{king}.
    By \lemref{kingfindmatch} and \lemref{kingfindtails}, we conclude
    \partinref{eraw}{ifftails}. Also by \lemref{kingfindmatch} we can guarantee \partinref{eraw}{composition}.

    To show \partinref{eraw}{timebound}, we observe that the worst-case runtime
    for a single execution of \lineref{callexact}{eraw} happens when the lower
    link of $v$ is of size $\Theta(n/2)$.  Computing the optimal pairings that
    \extract\ returns is at least as hard as computing the homology of $K'$,
    which is of the time complexity of matrix multiplication.
    By~\cite{raz2002complexity}, we know that the runtime of \eraw\ is
    lower-bounded by~\erawrun.
\end{proof}

%% file: body/append-defs.tex
\section{Equivalence of Definitions}\label{append:defs}

We gave the following definition of discrete Morse function:

\morsedefn*

However, in \cite{nic}, proves the following:

\brad{this is problem 2.23 in \cite{nic}}
\begin{lemma}[Regular Characterization]\label{lem:regsim}
A $p-$simplex $\sigma$ is regular if and only if either of the 
following holds
\begin{enumerate}[(i)]
\item There exists $\tau^{(p+1)}\succ \sigma$ such that $f(\tau)\leq f(\sigma)$.
\item There exists $\gamma^{(p-1)}\prec \sigma$ such that $f(\gamma)\geq f(\sigma)$.
\end{enumerate}

\end{lemma}
In fact both of the above properties can not both be true.
As shown in:
\begin{lemma}[Exclusion]\label{lem:exclusion}
Let $f:K\rightarrow \R$ be a discrete Morse
function and $\sigma\in K$ a regular simplex.
Then conditions $1)$ and $2)$ of \lemref{regsim}
cannot both be true, Hence, exactly one of the conditions hold
whenever $\sigma$ is regular. 
\end{lemma}

Regular pairs can not differ by more than one dimension.

\brad{also stated as a problem in \cite{nic}}
\begin{lemma}[Impossible Pairs]\label{lem:impossible}
Let $f:K\rightarrow \R$ be a discrete Morse
function. it is impossible
to have a pair of simplices $\tau^i\prec \sigma^p$ in $K$
with $i<p-1$ such that $f(\tau)>f(\sigma).$

\end{lemma}

The previous lemmas give the following characterization of
a regular simplex in discrete Morse function:

\begin{lemma}[Check if Regular]\label{lem:dmf-equiv}
   A simplex $\sigma$ is regular if and only if $\sigma$ is paired
   with either a face or a coface with co$-$dimension 1 but not both.
   
\end{lemma}

\begin{proof}
    \todo{}
\end{proof}

\todo{now, we can say that the definitions of critical are the same as well.}

%% file: body/append-link.tex
\section{Computing the Link}

\brittany{Someone (maybe Dave?) needs to go carefully through the stuff that is
thrown into this file.  Things that I know we'll need:
(1) we reference \lemref{rllbound} from within \thmref{king}.
(2) We need to call RecursiveLowerLink from \algref{eraw} or something like that.
}

\input{body/lower-link}

\begin{algorithm}
    \caption{Link}\label{alg:link}
    \begin{algorithmic}[1]
    \Require $K$ a \simp\  with $n_0$ vertices, $f$ an injective map from $K$ to $\R,$ the sorted Hasse diagram, $\mathcal{H}_K^*$ and a vertex $v.$
    \Ensure The lower link of $v.$

    \State{Find all cofaces of $v$, color these simplices blue.}
    \State{Find all faces of the blue simplices, color them red if they are not blue.}

    \Return the red vertices.
    \end{algorithmic}
\end{algorithm}

\begin{algorithm}
    \caption{\rllink}\label{alg:recursivell}
    \begin{algorithmic}[1]
        \Require $K$ a \simp\  with $n_0$ vertices, $f$ an injective map from $K$ to $\R.$
        \Ensure The number lower link computations in \eraw.
        \State Let $\ell = 0$
        \ForAll{$v \in K$} \label{line:recursivell-rllv}
            \State Let $K' :=$ the lowerlink of $v$ and let $f'$ be $f$ restricted to $K'.$
            \State $\ell = \ell +1$

            \If{$K'\neq \emptyset$}

            \State \rllink($K',f'$)

            \EndIf
        \EndFor
        \State Return $\ell.$
    \end{algorithmic}
\end{algorithm}

We count the number of times the lower link of a vertex is computed in
\algref{recursivell} when the input
contains a $d$-dimensional simplex.

Assuming that computing the lower link is linear in the number of simplices,
since there are $\Omega(2^{d+1})$ lower link computations and each take
$\Theta(n)$ time, we have the following:

\begin{lemma}[Recursive Lower Link Bound]\label{lem:rllbound}
    The number of lower link computations in \rllink\ is $\Omega(2^{d+1}).$
\end{lemma}

\begin{proof}
    We use strong induction on $d.$ Let $L_d$ be the number of lower link
    computations of a $d-$simplex in \rllink. When $d=0$ we have a
    single vertex and we compute one lower link.

     Suppose the
    vertices have been sorted by function value, then when $v_{0}$ is selected
    in \lineref{rllv}{recursivell}, the lower link is empty, when $v_{1}$ is selected the
    lower link is a single vertex, when $v_{2}$ is selected the lower link is an
    edge. This pattern continues.  When $v_{i}$ is
    selected the lower link is a $i-1-$simplex. When we iterate over all
    vertices in $\sigma$ we call \rllink\ on a simplex of every
    dimension from $0$ to $d-1.$ Since there are $d+1$ vertices in $\sigma$ we
    have an additional $d+1$ lower link computations.

    This gives the following recurrence relation $$L_d=L_{d-1} +L_{d-2}+\ldots +
    L_{1}+L_0+(d+1)$$ with $L_0=1.$ Strong induction shows that the number of
    lower link computations is $2^{d+1}-1.$
\end{proof}


\eraw\ proceeds recursively on the lower link of each vertex.
This leads to many unnecessary computations of the lower link of vertices.
\shortlineref{llcall}{eraw} and \shortlineref{callexact}{eraw} of \eraw\ recursively compute the lower link
of each vertex in $K.$ We include this subalgorithm, called
\rllink, in \algref{recursivell}.

%% file: body/lower-link.tex
We now give an algorithm that computes the link of a vertex in a simplicial complex. Like our main algorithm,
the link algorithm can be visualized using the Hasse diagram. For a given vertex, $v$ assign all simplicies
that contain $v$ the color blue.
Then assign all faces
of the blue simplicies the color red if they are not already blue. The lower link of $v$ is all simplicies that are red and not blue and have
value less than the value of $v.$ See \figref{links} and \algref{link}.
Notice that a simplex, $\sigma,$ is red if it is contained in a simplex that contains $v$ but $v$ is not
contained in $\sigma,$ the definition of link$(v).$

\begin{figure}[htb]
        \centering
        \begin{subfigure}[b]{0.35\textwidth}
        \includegraphics[width=\textwidth]{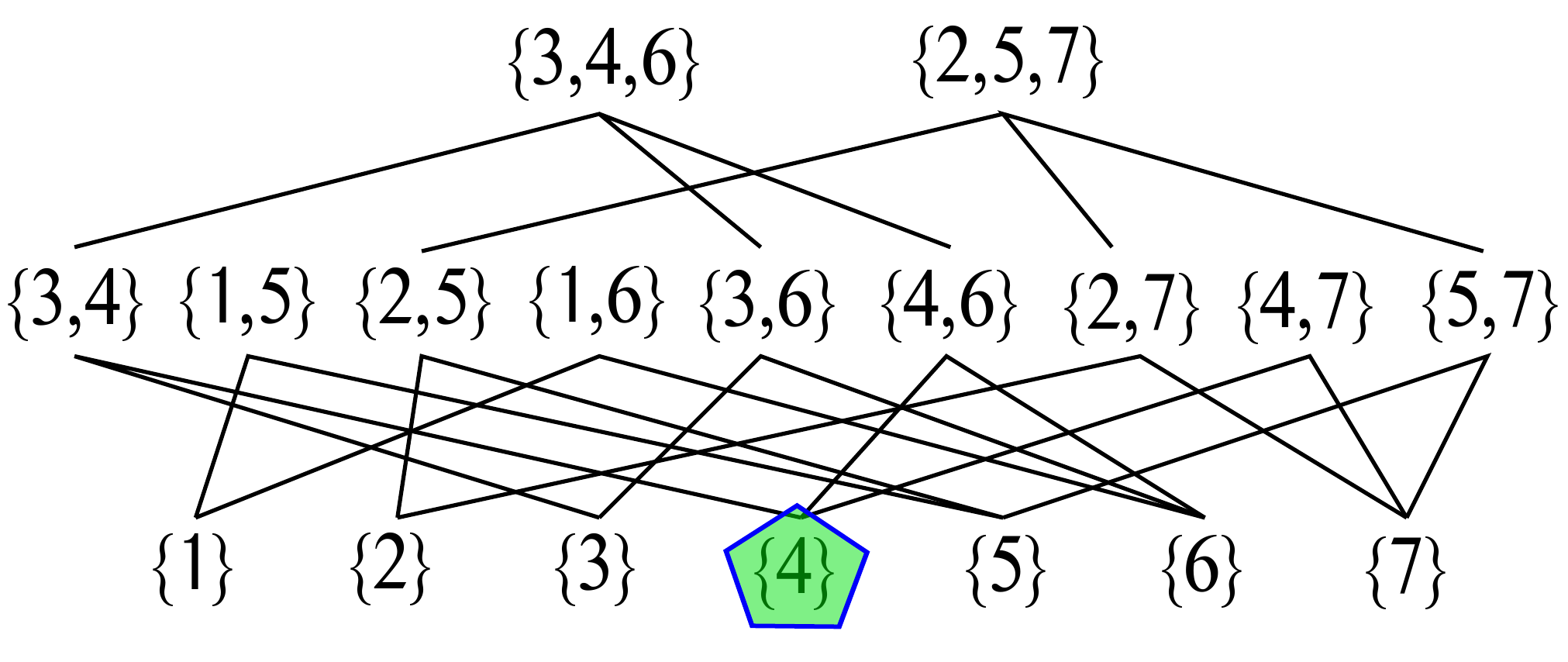}
        \caption{Say we wish to find the link of the vertex with function value $4$ in \figref{cat}.}
	 \label{fig:link1}
        \end{subfigure}
        \hspace{1cm}
        \begin{subfigure}[b]{0.35\textwidth}
        \includegraphics[width=\textwidth]{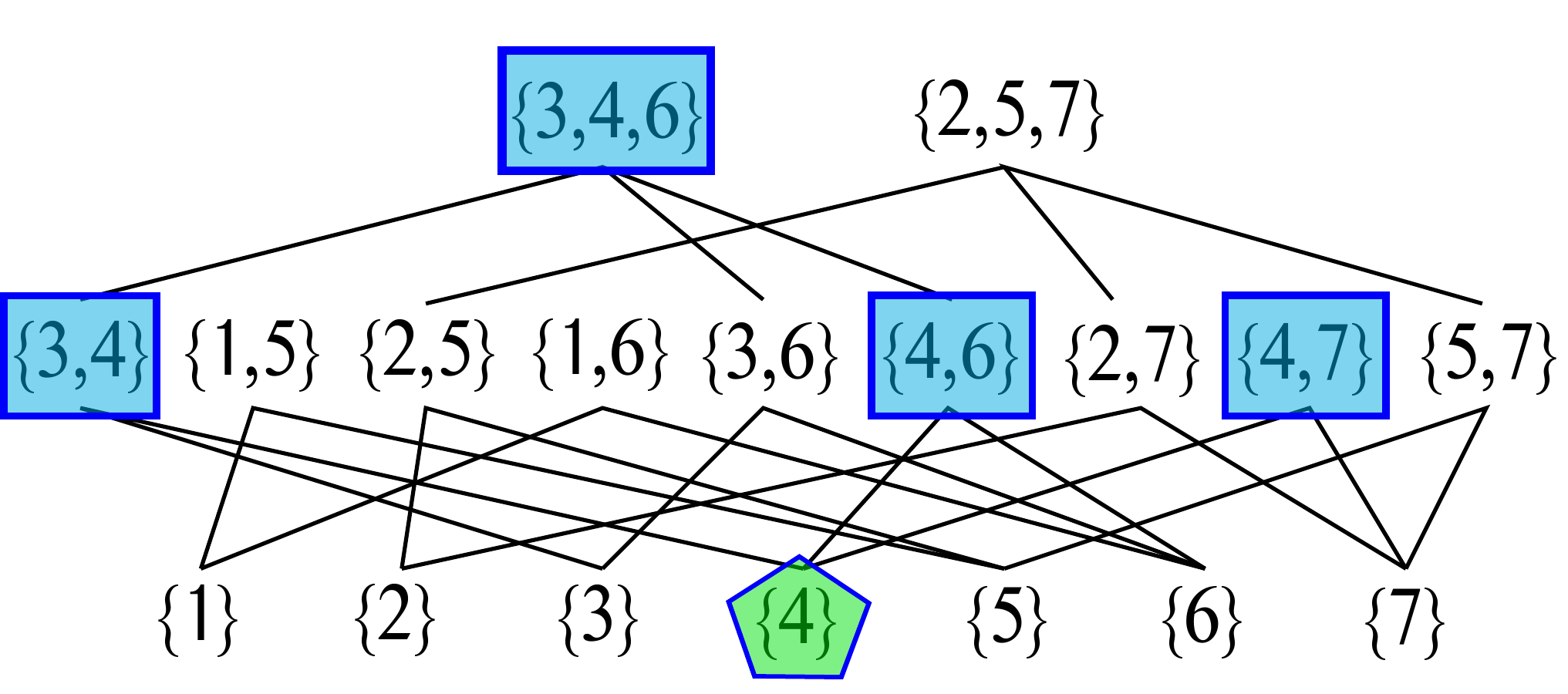}
        \caption{We move up the Hasse diagram and find all simplices containing $4$
		and color them blue.}
\label{fig:link2}
        \end{subfigure}

        \begin{subfigure}[b]{0.35\textwidth}
        \includegraphics[width=\textwidth]{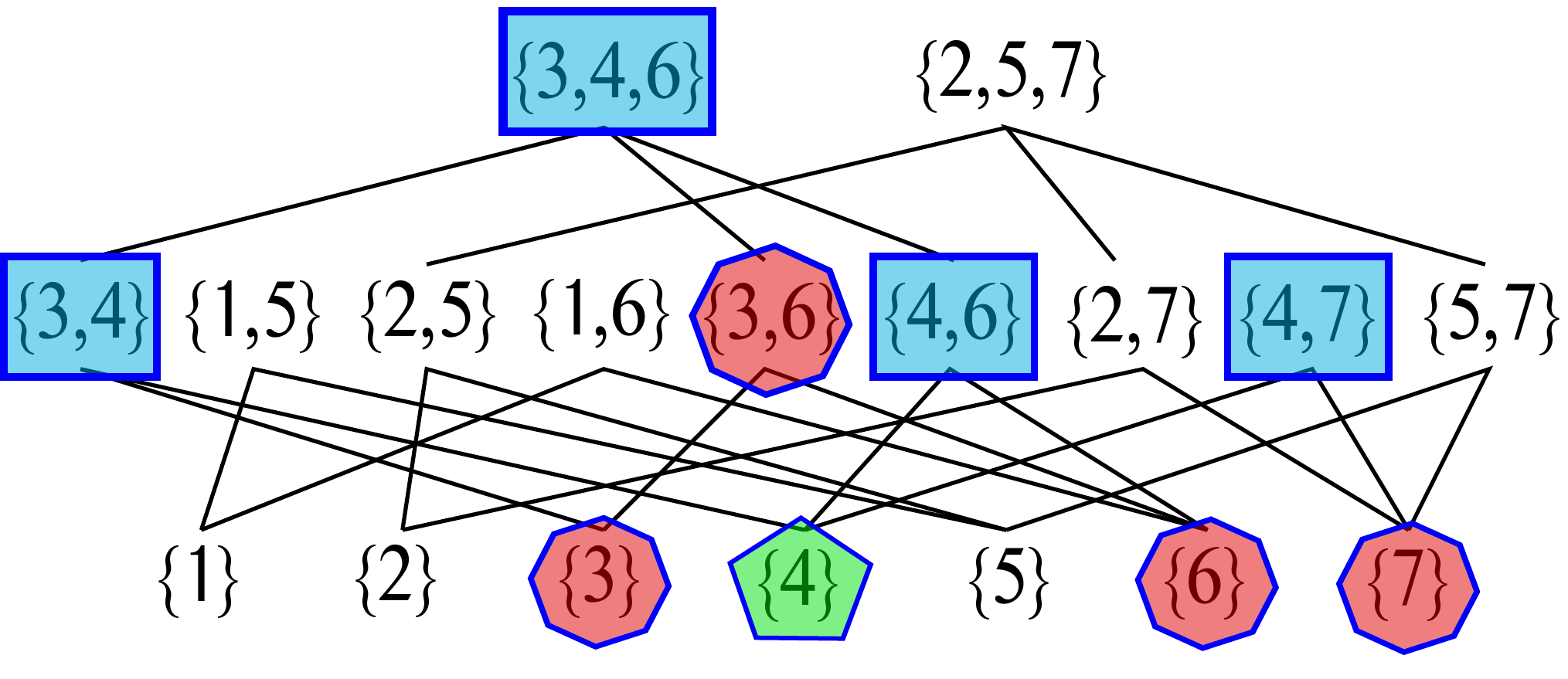}
        \caption{Next, we find all faces of the blue simplicies that are not already blue. These red simplices
		are the link of $4$.}
\label{fig:link3}
        \end{subfigure}
		\caption{Computing the Link of a Vertex.}
			\label{fig:links}
\end{figure}

\begin{lemma}[Lower Link]

The lower link of each vertex can be computed in $O(n).$
\label{lower}
\end{lemma}
\begin{proof}
We iterate over everything in the Hasse diagram at most twice.
\end{proof}

%% file: cccg20.bbl
\begin{thebibliography}{10}

\bibitem{uli-phd}
U.~Bauer.
\newblock {\em Persistence in Discrete Morse Theory}.
\newblock PhD thesis, Nieders{\"a}chsische Staats-und
  Universit{\"a}tsbibliothek G{\"o}ttingen, 2011.

\bibitem{bauer2012optimal}
U.~Bauer, C.~Lange, and M.~Wardetzky.
\newblock Optimal topological simplification of discrete functions on surfaces.
\newblock {\em Discrete and Computational Geometry}, 47(2):347--377, 2012.

\bibitem{bauer-rathod-18}
U.~Bauer and A.~Rathod.
\newblock Hardness of approximation for {M}orse matching.
\newblock arXiv:1801.08380, 2018.

\bibitem{vcomic2011dimension}
L.~{\v{C}}omi{\'c} and L.~De~Floriani.
\newblock Dimension-independent simplification and refinement of {M}orse
  complexes.
\newblock {\em Graphical Models}, 73(5):261--285, 2011.

\bibitem{wang18}
T.~Dey, J.~Wang, and Y.~Wang.
\newblock Graph reconstruction by discrete {M}orse theory.
\newblock In {\em 34th Symposium on Computational Geometry (SoCG)}, pages
  31:1--31--13, 2018.

\bibitem{edelsbrunner2010computational}
H.~Edelsbrunner and J.~Harer.
\newblock {\em Computational Topology: {A}n Introduction}.
\newblock American Mathematical Society, 2010.

\bibitem{edelsbrunner2003hierarchical}
H.~Edelsbrunner, J.~Harer, and A.~Zomorodian.
\newblock Hierarchical {M}orse-{S}male complexes for piecewise linear
  $2$-manifolds.
\newblock {\em Discrete and Computational Geometry}, 30(1):87--107, 2003.

\bibitem{edelsbrunner02}
H.~Edelsbrunner, D.~Letscher, and A.~Zomorodian.
\newblock Topological persistence and simplification.
\newblock {\em Discrete and Computational Geometry}, 28:511–533, 2002.

\bibitem{forman98}
R.~Forman.
\newblock Discrete {M}orse theory for cell complexes.
\newblock {\em Advances in Mathematics}, 134:90--145, 1998.

\bibitem{forman02}
R.~Forman.
\newblock A user's guide to discrete {M}orse theory.
\newblock {\em S\'eminaire Lotharingien de Combinatoire}, 42:Art. B48c, 35pp,
  2002.

\bibitem{hersh05}
P.~Hersh.
\newblock On optimizing discrete {M}orse functions.
\newblock {\em Advances in Applied Math}, 35:294--322, 2005.

\bibitem{joswig04}
M.~Joswig and M.~Pfetsch.
\newblock Computing optimal {M}orse matchings.
\newblock {\em SIAM Journal on Discrete Mathematics (SIDMA)}, 20(1):11--25,
  2006.

\bibitem{king}
H.~King, K.~Knudson, and N.~Mramor.
\newblock Generating discrete {M}orse functions from point data.
\newblock {\em Experimental Mathematics}, 14:435--444, 2005.
\newblock MR2193806.

\bibitem{knudson2015morse}
K.~Knudson.
\newblock {\em Morse Theory: Smooth and Discrete}.
\newblock World Scientific Publishing Company, 2015.

\bibitem{lewiner03}
T.~Lewiner, H.~Lopes, and G.~Tavares.
\newblock Toward optimality in discrete {M}orse theory.
\newblock {\em Experimental Mathematics}, 12:271--285, 2003.

\bibitem{milnor63}
J.~Milnor.
\newblock {\em Morse Theory}.
\newblock Princeton University Press, Princeton, New Jersey, 1963.

\bibitem{nanda2013morse}
K.~Mischaikow and V.~Nanda.
\newblock Morse theory for filtrations and efficient computation of persistent
  homology.
\newblock {\em Discrete and Computational Geometry}, (50):330, 2013.

\bibitem{raz2002complexity}
R.~Raz.
\newblock On the complexity of matrix product.
\newblock {\em SIAM Journal on Computing}, 32:1356–1369, 2003.

\bibitem{nic}
N.~Scoville.
\newblock {\em Discrete Morse Theory}.
\newblock American Mathematical Society, Providence, Rhode Island, 2019.

\end{thebibliography}
